\documentclass{llncs}

\usepackage{amssymb} 
\usepackage{mathtools} 
\usepackage[usenames, dvipsnames]{color} 
\usepackage{caption} 
\definecolor{pink1}{rgb}{0.858, 0.188, 0.478}

\usepackage{datetime}
\usepackage{wrapfig}
\usepackage[normalem]{ulem}  

\newcommand{\e}{\mathrm{e}}
\newcommand{\Prob}[1][ ]{P #1} 
\newcommand{\N}[1]{\mathcal{N}(#1)} 

\renewcommand{\|}{\,||\,}
\newcommand{\QED}{\quad \blacksquare}

\begin{document}

\pagestyle{plain} 

\setlength{ \intextsep}{3pt}
\setlength{\abovedisplayskip}{3pt}
\setlength{\belowdisplayskip}{3pt}
\setlength{\topsep}{2pt}
\setlength{\textfloatsep}{0.1cm}

\title{Chemical Boltzmann Machines}
\titlerunning{Chemical Boltzmann Machines}  
%
\author{William Poole\inst{1}* \and Andr\'es Ortiz-Mu\~noz\inst{1}* \and Abhishek Behera\inst{2}* \and Nick S. Jones\inst{3} \and \\ Thomas E. Ouldridge\inst{3} \and Erik Winfree\inst{1} \and Manoj Gopalkrishnan\inst{2} }
\authorrunning{William Poole et al.} 
%
\tocauthor{William Poole, Andr\'es Ortiz-Mu\~noz, Abhishek Behera, Nick S. Jones, Thomas E. Ouldridge, Erik Winfree, and Manoj Gopalkrishnan}

\institute{California Institute of Technology, Pasadena, CA\\
\email{wpoole@caltech.edu},
\and
India Institute of Technology Bombay, Mumbai, India \\
\email{manoj.gopalkrishnan@gmail.com}
\and
Imperial College London, London, England \\
\email{t.ouldridge@imperial.ac.uk}
}


\maketitle 
{\small \centering *Contributed Equally \\

}
%
\vspace{-.4cm}

\begin{abstract}
How smart can a micron-sized bag of chemicals be? How can an artificial or real cell make inferences about its environment? From which kinds of probability distributions can chemical reaction networks sample? We begin tackling these questions by showing four ways in which a stochastic chemical reaction network can implement a Boltzmann machine, a stochastic neural network model that can generate a wide range of probability distributions and compute conditional probabilities. The resulting models, and the associated theorems, provide a road map for constructing chemical reaction networks that exploit their native stochasticity as a computational resource. Finally, to show the potential of our models, we simulate a chemical Boltzmann machine to classify and generate MNIST digits in-silico.
\end{abstract}

\vspace{-.6cm}
\section{Introduction}
\vspace{-.2cm}

To carry out complex tasks such as finding and exploiting food sources, avoiding toxins and predators, and transitioning through critical life-cycle stages, single-celled organisms and future cell-like artificial systems must make sensible decisions based on information about their environment ~\cite{Bray1995,Bray2009}.  The small volumes of cells makes this enterprise inherently probabilistic:  environmental signals and the biochemical networks within the cell are noisy, due to the stochasticity inherent in the interactions of small, diffusing molecules  ~\cite{McAdams1997,Elowitz2002,Perkins2009}.  The small volumes of cells raises questions not only about how stochasticity influences circuit function, but also about how much computational sophistication can be packed into the limited available space. 

Perhaps surprisingly, neural network models provide an attractive architecture for the types of computation, inference, and information processing that cells must do.  Neural networks can perform deterministic computation using circuits that are smaller and faster than boolean circuits composed of AND, OR, and NOT gates~\cite{Muroga1971}, can robustly perform tasks such as associative recall~\cite{Hopfield1982}, and naturally perform Bayesian inference~\cite{Hinton1984}.  Furthermore, the  structure of biochemical networks, such as signal transduction cascades~\cite{Bray1990,Bray1995,Hellingwerf1995} and  genetic regulatory networks~\cite{Mjolsness1991,Mestl1996,Buchler2003,Deutsch2014,Deutsch2016}, can map surprisingly well onto neural network architectures.  Chemical implementations of neural networks and related machine learning models have also been proposed~\cite{Hjelmfelt1991,Hjelmfelt1992,Kim2004,Napp2013,Gopalkrishnan2016}, and limited examples demonstrated~\cite{Hjelmfelt1993,Kim2006,Kim2011,Qian2011}, for synthetic biochemical systems.

\begin{figure}[tb!!!]
\includegraphics[width=.95\textwidth]{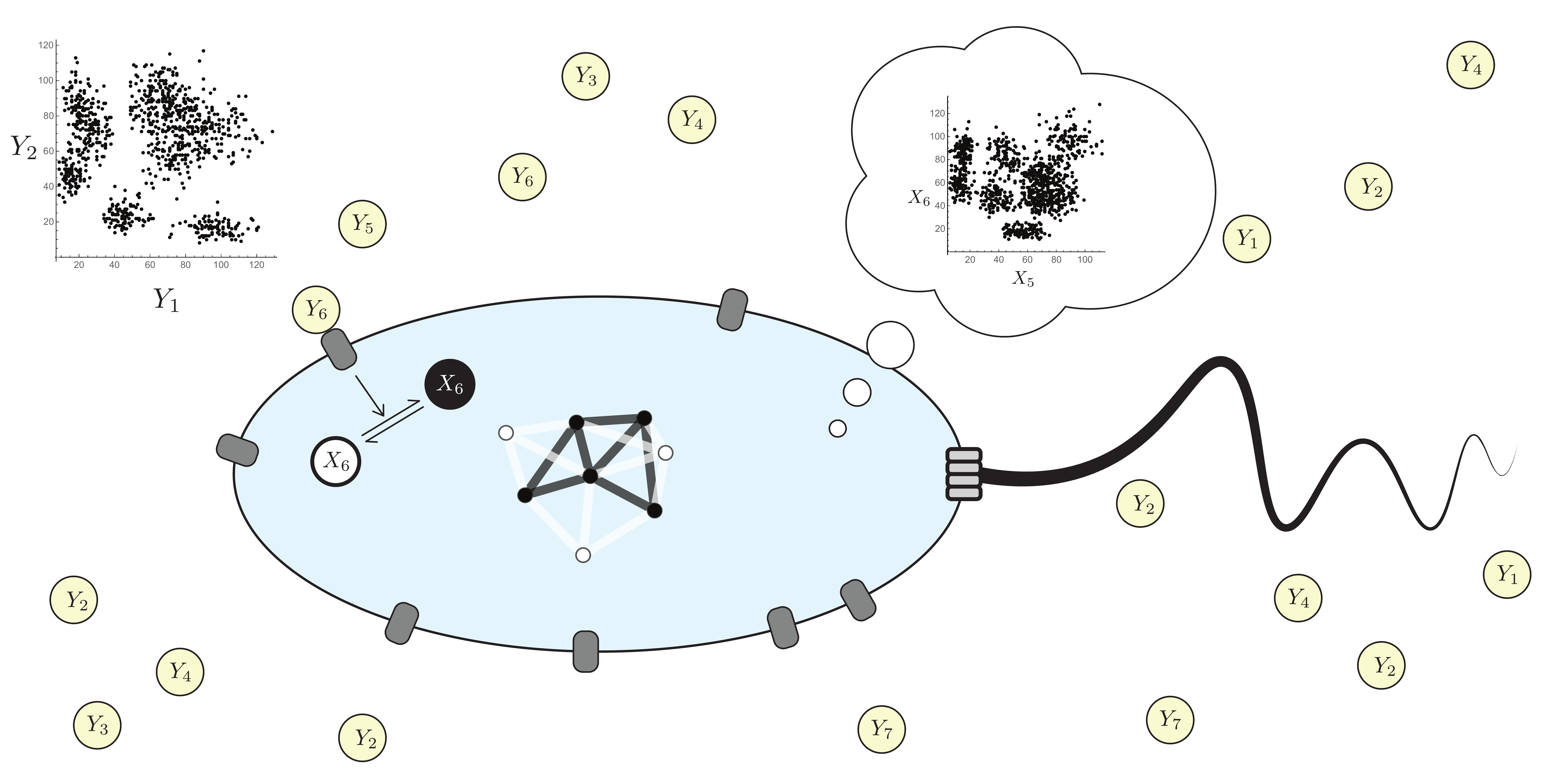}
\caption{In a micron-scale environment, molecular counts are low and a real (or synthetic) cell will have to respond to internal and environmental cues. Probabilistic inference using chemical Boltzmann machines provides a framework for how this may be achieved.}
\end{figure}

Previous work on biochemical neural networks and biochemical inference invoked models based on continuous concentrations of species representing neural activities.  Such models are limited in their ability to address questions of biochemical computation in small volumes, where
discrete stochastic chemical reaction network  models  must be used to account for the low molecular counts.  The nature of biochemical computation changes qualitatively in this context.
%
In particular, stochasticity has been widely studied in genetic regulatory networks \cite{Lestas2008}, signaling cascades \cite{Lestas2010}, population level bet hedging in bacteria~\cite{Veening2008}, and other areas \cite{Balazsi2011,Tsimring2014}
-- where the stochasticity is usually seen 
as a challenge limiting correct function, but is occasionally
also viewed as a useful resource~\cite{Eldar2010}. 
Our work falls squarely in the latter camp: we
attempt to exploit the intrinsic stochastic fluctuations of a formal chemical reaction network (CRN) to build natively stochastic samplers by implementing a stochastic neural network.
This links to efforts to build natively stochastic hardware for Bayesian inference~\cite{Mansinghka2009,Wang2016} and to the substantial literature attempting to model, and find evidence for, stochastic neural systems capable of Bayesian inference \cite{Fiser2010,Pouget2013}. 



Specifically, we propose CRNs that implement Boltzmann machines (BMs), a flexible class of Markov random fields capable of generating diverse distributions and for which conditioning on data has straightforward physical interpretations \cite{Hinton1984,Ackley1985}. BMs are an established model of probabilistic neural networks due to their analytic tractability and identity with spin systems in statistical physics \cite{Tanaka1998} and Hopfield networks in computer science \cite{Hopfield1982}. These networks have been studied extensively and used in a wide range of applications including image classification \cite{Tang2011} and video generation \cite{Taylor2009}. 
We prove that CRNs can implement BMs and that this is possible using detailed balanced CRNs. Moreover, we show that many of the attractive features of BMs can be applied to our CRN constructions such as inference, a straightforward learning rule and scalability to real-world data sets. We thereby introduce the idea of a chemical Boltzmann machine (CBM), a chemical system capable of exactly or approximately performing inference using a stochastically sampled high-dimensional state space, and explore some of its possible forms.

\vspace{-.2cm}
\section{Relevant Background}
\vspace{-.2cm}

\subsection{Boltzmann Machines (BMs):}
\vspace{-.1cm}
\label{sec:BM}
Boltzmann machines are a class of binary stochastic neural networks, meaning that each node randomly switches between the values 0 and 1 according to a specified distribution. They are widely used for unsupervised machine learning because they can compactly represent and manipulate high-dimensional probability distributions. Boltzmann machines provide a flexible machine learning architecture because, as generative models, they can be used for a diverse set of tasks including  data classification, data generation, and data reconstruction. Additionally, the simplicity of the model makes them analytically tractable. The use of hidden units (described below) allows Boltzmann machines to represent high order correlations in data. Together, these features make Boltzmann machines an excellent starting point for implementing stochastic chemical computers.

Fix a positive integer $N\in\mathbb{Z}_{>0}$. An $N$-node \textbf{Boltzmann machine} (BM) is specified by
a quadratic \textbf{energy} function $E:\{0,1\}^N\to\mathbb{R}$
\begin{equation}
E(x_1,x_2,\dots,x_N)=-\sum_{i<j}w_{ij}x_ix_j - \sum_i \theta_i x_i
\end{equation}
where $\theta_i\in\mathbb{R}$ is the \textbf{bias} of node $i$, and $w_{ij}=w_{ji}\in\mathbb{R}$ is the \textbf{weight} of the unordered pair $(i,j)$ of nodes, with $w_{ii} = 0$.  One may specify a BM \textbf{architecture}, or graph topology, by choosing additional weights $w_{ij}$ that are to be set to $0$. In this paper, we will use $\N{i} = \{j \textrm{ s.t. } w_{ij} \neq 0\}$ to denote the neighborhood of $i$. From a physical point of view, we are implicitly using {\it{temperature units}} $k_BT$ for energy, which we will continue to do throughout this paper. A BM describes a distribution $\Prob{(x)}$ over \textbf{state vectors} $x=(x_1,\dots, x_N)\in\{0,1\}^N$,
\begin{equation}
\label{BM Dist}
\Prob{(x)}=\frac{1}{Z}e^{-E(x)} \quad \textrm{with} \quad Z = \sum_{x' \in \{0, 1\}^N} e^{-E(x')}.
\end{equation}

Nodes of a BM are often partitioned into sets $V$ and $H$  of \textbf{visible} and \textbf{hidden}, respectively. Nodes in $V$ represent data, and auxiliary nodes in $H$ allow more complex distributions to be represented in the visible nodes. An \textbf{implementation} of a BM is a stationary stochastic process that samples from this distribution in the steady state. A BM can be implemented {\it in silico} using the \textbf{Gibbs sampling} algorithm \cite{Casella1992}, which induces a discrete time Markov chain (DTMC) on the state space $\{0,1\}^N$ in such a way that the stationary distribution of this Markov chain corresponds to the distribution $\Prob{(x)}$. In each round, one node $i\in\{1,\dots, n\}$ is chosen at random for update. For any two adjacent configurations $x$ and $x'$ which differ only at node $i$ --- i.e., $x_i = 1-x'_i$ and $x_j=x'_j$ for all $j\neq i$ --- we set the transition probabilities $T_{x \to x'}$ of the DTMC so that
\begin{equation}
\label{BMDetailedBalance}
\frac{T_{x' \to x}}{ T_{x \to x'}}
= \frac{\Prob{(x)}}{\Prob{(x')}} = \frac{\e^{-E(x)}}{\e^{-E(x')}} = \e^{\left(\theta_i + \sum_{j \in \N{i}} w_{ij} x_j \right)(x_i-x'_i)}.
\end{equation}
Any function $T_{x \to x'}$ can be chosen so long as (\ref{BMDetailedBalance}) is satisfied.
One common choice is $T_{x \to x'} = 1/(N (1+\e^{E(x')-E(x)}))$.

A Boltzmann machine is also an inference engine. One can do inference on $\Prob{(x)}$ by conditioning on the values of a subset of the nodes. Suppose nonempty node subsets $U$ and $Y$ form a partition of the nodes $\{1,2,\dots,N\}$, and fix $u\in\{0,1\}^U$. To obtain samples from $\Prob{(y \mid u)}$ where $y\in\{0,1\}^Y$, one \textbf{clamps} every node $i\in U$ to the state $u_i$ while running Gibbs sampling, i.e., one does not allow these nodes to update. Clamping nodes to an input state is the same as specifying the input data for a statistical model. Steady state samples $y\in\{0,1\}^Y$ of this procedure are draws from the distribution $\Prob{(y\mid u)}$.

Boltzmann machines can be used to learn a generative model from unlabeled data. After specifying the architecture, one then proceeds to find the weights, $w_{ij}$, and biases, $\theta_i$, that maximize the likelihood of the observed data according to the model, using gradient descent from a random initial parameterization. This reduces to a very simple two-phase learning rule  where weights on active edges are strengthened in a ``wake phase'' during which the network is clamped to observed data and  are weakened in a ``sleep phase'' during which the network runs free~\cite{Hinton1984,Ackley1985}.  Given a target distribution $Q(x)$, this gradient descent corresponds to calculating the gradient of the Kullback-Leibler divergence from $\Prob$ to $Q$, 
$D_{KL}(Q\|\Prob) = \sum_x Q(x) \log \frac{Q(x)}{\Prob(x)}$, 
with respect to the parameters $\theta_i$ and $w_{ij}$:
\begin{equation}
\frac{d \theta_i}{dt}=
-\frac{\partial D_{KL}}{\partial 
\theta_i} = \langle x_i \rangle_{Q} - 
\langle x_i \rangle_{\Prob}
\quad \textrm{and} \quad
\frac{d w_{ij}}{dt} =
-\frac{\partial D_{KL}}{\partial w_{ij}} = \langle x_i x_j \rangle_{Q} - \langle x_i x_j \rangle_{\Prob}
\label{eq:Hebbian}
\end{equation}
where $\langle \cdot \rangle_{\Prob}$ and $\langle \cdot \rangle_Q$ 
denote expected values with respect to the 
distributions $\Prob$ and $Q$ respectively.
When hidden units are present, the distribution $Q$ (which is defined on visible units only) is extended to hidden units based on clamping the visible units according to $Q$ and using the conditional distribution $P(y|u)$ to determine the hidden units.

\vspace{-.2cm}
\subsection{Chemical Reaction Networks (CRNs):}
\vspace{-.1cm}

Fix a finite set $\mathcal{S}=(S_1, S_2, \dots, S_M)$ of $M$ \textbf{species}. A \textbf{reaction} $r$ is a formal chemical equation
\begin{equation}
\sum_{i = 1}^M 
\mu_r^i S_i \rightarrow \sum_{i=1}^M \nu_r^i S_i,
\end{equation}
abbreviated as $r=\mu_r \to \nu_r$ where $\mu_r,\nu_r\in \mathbb{N}^\mathcal{S}$ are the stoichiometric coefficient vectors for the reactant and product species respectively, and $\mathbb{N} = \mathbb{Z}_{\geq 0}$. A \textbf{reaction rate constant}, $k_r \in \mathbb{R}_{>0}$, is associated with each reaction. In this paper, we define a \textbf{chemical reaction network} (CRN) as a triple $\mathcal{C} = (\mathcal{S}, \mathcal{R}, k)$ where $\mathcal{S}$ is a finite set of
species, and $\mathcal{R}$ is a set of reactions, and $k$ is the associated set of reaction rate constants.

We will denote chemical species by capital letters, and their counts by lower case letters; e.g., $s_1$ denotes the number of species $S_1$. Thus the state of a stochastic CRN is described by a vector on a discrete lattice, $s = (s_1, s_2 \dots s_M ) \in \mathbb{N}^\mathcal{S}$. 
The dynamics of a stochastic CRN are as follows~\cite{gillespie2007stochastic}.  The probability that a given reaction occurs per unit time, called its \textbf{propensity}, is given by 
\begin{equation}
\rho_r(s) = k_r\prod_{i=1}^M \frac{s_i!}{(s_i - \mu^i_r)!}
\quad \textrm{if} \quad s_i \ge \mu^i_r \quad \textrm{and $0$ otherwise}.
\label{eq:CRN2}
\end{equation}
Each time a reaction fires, state $s$ changes to state $s+\Delta_r$, where $\Delta_r = \nu_r - \mu_r$ is called the reaction vector, and the propensity of each reaction may change.
Viewed from a state space perspective, stochastic CRNs are continuous time Markov chains (CTMCs) with transition rates 
\begin{equation}
R_{s \to s'} = \sum_{r \textrm{\ s.t.\ } s' = s + \Delta_r} \rho_r(s)
\end{equation}
and thus their dynamics follow
\begin{equation}
\frac{d\Prob(s,t)}{d t} = 
\sum_{s' \neq s} R_{s' \to s} P(s',t) - R_{s \to s'} P(s,t)\ \ ,
\end{equation}
where $\Prob(s,t)$ is the probability of a state with counts $s$ at time $t$.
Equivalently, they are
governed by the \textbf{chemical master equation}, 
\begin{equation}
\frac{d\Prob(s,t)}{d t} = 
\sum_{r \in \mathcal{R}}\Prob(s-\Delta_r,t)\rho_r(s-\Delta_r)-\Prob(s,t)\rho_r(s) \ \ .
\end{equation}
A stationary distribution $\pi(s)$ may be found by solving $\frac{d \Prob(s,t)}{d t} = 0$ simultaneously for all $s$; in general, it need not be unique, and even may not exist. Given an initial state $s_0$, $\pi(s)=P(s,\infty)$ is unique if it exists. For that initial state, the \textbf{reachability class} $\Omega_{s_0} \subseteq \mathbb{N}^M $ is the maximal subset of the integer lattice accessible to the CRN via some sequence of reactions in $\mathcal{R}$. 
We will specify a CRN and a reachability class given an initial state as a shorthand for specifying a CRN and a set of initial states with identical reachability classes.

\vspace{-.2cm}
\subsection{Detailed Balanced Chemical Reaction Networks:}
\vspace{-.1cm}

A CTMC is said to satisfy \textbf{detailed balance} if there exists a well-defined function of state $s$, $E(s) \in \mathbb{R}$, such that for every pair of states $s$ and $s'$, the transition rates $R_{s \to s'}$ and $R_{s' \to s}$ are either both zero or 
\begin{equation}
\frac{R_{s \to s'}}{R_{s' \to s}} = \e^{E(s)-E(s')} \ . 
\label{eq:dbcrn1}
\end{equation}
If the state space $\Omega$ is connected and the partition function $Z = \sum_{s \in \Omega} \e^{-E(s)}$ is finite, then the steady state distribution $\pi(s) = \frac{1}{Z} \e^{-E(s)}$ is unique, and the net flux between all states is zero in that steady state.

There is a related but distinct notion of detailed balance for a CRN.
An equilibrium chemical system is constrained by physics to obey detailed balance at the level of each reaction. In particular, for a dilute equilibrium system,  each species $S_i \in \mathcal{S}$ has an energy $G[S_i] \in \mathbb{R}$ that relates to its intrinsic stability, and 
\begin{equation}
\frac{k_{r^+}}{k_{r^-}} = \e^{-\sum_{i=1}^M \Delta_{r^+}^i G[S_i]} =  \e^{-\Delta G_{r^+}},
\label{eq:dbcrn2}
\end{equation} 
where $\Delta_{r^+}^i $ is the $i$th component of $\Delta_{r^+} = \nu_{r^+} - \mu_{r^+} $, and we have defined $\Delta G_{r^+} = \sum_{i=1}^N \Delta_{r^+}^i G[S_i]$. 
Any CRN for which there exists a function $G$ satisfying (\ref{eq:dbcrn2}) is called a detailed balanced CRN.
To see that the CTMC for a detailed balanced CRN also itself satisfies detailed balance in the sense of (\ref{eq:dbcrn1}),
let $s' = s + \Delta_{r+}$ and
note that (\ref{eq:CRN2}) and (\ref{eq:dbcrn2}) imply that
\begin{equation}
\label{eq:dbcrn3}
\frac{\rho_{r^+}(s)}{\rho_{r^-}(s')} = \e^{\mathcal{G}(s)- \mathcal{G}(s')}
\quad \textrm{with} \quad
\mathcal{G}(s) = \sum_{i=1}^{M} s_i G[S_i] + \log (s_i!),
\end{equation}
for all reactions $r^+$. Here, $\mathcal{G}(s)$ is a well-defined function of state $s$ (the free energy) that can play the role of $E$ in (\ref{eq:dbcrn1}). If there are multiple reactions that bring $s$ to $s'$, they all satisfy (\ref{eq:dbcrn3}), and therefore  the ratio $R_{s \to s'}/R_{s' \to s}$ satisfies (\ref{eq:dbcrn1}) and the CTMC satisfies detailed balance.

It is possible to consider non-equilibrium CRNs that violate (\ref{eq:dbcrn2}). Such systems must be coupled to implicit  reservoirs of fuel molecules that can drive the species of interest into a non-equilibrium steady state \cite{Qian2007,Beard2008,Ouldridge2017}. Usually -- but not always \cite{joshi2013detailed,Erez2017} -- the resultant Markov chain violates detailed balance. In Section~\ref{sec:constructions1}, we shall consider a system that exhibits detailed balance at the level of the Markov chain, but is necessarily non-equilibrium and violates detailed balance at the detailed chemical level. 

Given an initial condition $s_0$, a detailed balanced CRN will be confined to a single reachability class $\Omega_{s_0}$.  Moreover, from the form of $\mathcal{G}(s)$,  the stationary distribution $\pi(s)$ on $\Omega_{s_0}$ of any detailed balanced CRN exists, is unique, and is a product of Poisson distributions restricted to the reachability class \cite{Anderson2010}, 
\begin{equation}
\label{ProductPoissonThrm}
\pi(s) = \frac{1}{Z}\e^{-\mathcal{G}(s)} = \frac{1}{Z}\prod_{i=1}^M \frac{\e^{-s_i G[S_i]}}{s_i!},
\end{equation}
with the partition function 
$
Z = \sum_{s' \in \Omega_{s_0}}  \e^{-\mathcal{G}(s')}
$
dependent on the reachability class. 
Note that this implies that the partition function is always finite, even for an infinite reachability class.

\vspace{-.2cm}
\section{Exact Constructions and Theorems}
\vspace{-.1cm}
\subsection{Clamping and Conditioning with Detailed Balanced CRNs:}
\vspace{-.1cm}
\label{sec:clamping_db_crns}

In a Boltzmann machine that has been trained to generate a desired probability distribution when run, inference can be performed by freezing, also known as clamping, the values of known variables, and running the rest of the network to obtain a sample; this turns out to exactly generate the conditional probability.  A similar result holds for a subclass of detailed balanced CRNs that generate a distribution, for an appropriate notion of clamping in a CRN.  Imagine a ``demon'' that, whenever a reaction results in a change in the counts of one of the clamped species, will instantaneously change it back to its previous value.  If every reaction is such that either no clamped species change, or else every species that changes is clamped, then the demon is effectively simply ``turning off'' those reactions.
More precisely, consider a CRN, $\mathcal{C} = (\mathcal{S}, \mathcal{R}, k)$. We will partition the species into two disjoint groups $Y = \mathcal{S}_{free}$ and $U = \mathcal{S}_{clamped}$, where $\mathcal{S}_{free}$ will be allowed to vary and $\mathcal{S}_{clamped}$ will be held fixed. We will define free reactions, $\mathcal{R}_{free}$, as reactions which result in neither a net production nor a net consumption of any clamped species. Similarly, clamped reactions, $\mathcal{R}_{clamped}$ are reactions which change the counts of any clamped species. 
The clamped CRN will be denoted $\mathcal{C}|_{U=u}$ to indicate the the species $U_i \in U$ have been clamped to the values $u_i$. The clamped CRN is defined by $\mathcal{C}|_{U=u} = (\mathcal{S}, \mathcal{R}_{free}, k_{free})$, that is, the entire set of species along with the reduced set of reactions and their rate constants. In the clamped CRN it is apparent that the clamped species will not change from their initial conditions because no reaction in $\mathcal{R}_{free}$ can change their count. However, these clamped species may still affect the free species catalytically.  If the removed reactions, $\mathcal{R}_{clamped}$, never change counts of non-clamped species, then $\mathcal{C}|_{U=u}$ is equivalent to the action of the ``demon'' imagined above.


We use equation \ref{ProductPoissonThrm} to prove that clamping a detailed balanced CRN is equivalent to calculating a conditional distribution, and to show when the conditional distributions of a detailed balanced CRN will be independent. Together, these theorems provide guidelines for devising detailed balanced CRNs with interesting (non-independent) conditional distributions and for obtaining samples from these distributions via clamping.

We will need one more definition.  Let $\mathcal{C}$ be a detailed balanced CRN with reachability class $\Omega_{s_0}$ for some initial condition $s_0 = (u_0, y_0)$. Let $\Gamma_{s_0}$ be the reachability class of the clamped CRN $\mathcal{C}|_{U=u_0}$ with species $U$ clamped to $u_0$ and species $Y$ free. We say clamping \textbf{preserves reachability} if $\Omega_{s_0|U=u_0}^Y = \Gamma_{s_0}^Y$ where $\Omega^Y_{s_0|U=u_0} = \{ y \textrm{ s.t. } (u_0,y) \in \Omega_{s_0} \}$ and $\Gamma^Y_{s_0} = \{ y \textrm{ s.t. } (u_0,y) \in \Gamma_{s_0} \}$. In other words, clamping preserves reachability if, whenever a state $s = (u_0,y)$  is reachable from $s_0$ by any path in $\mathcal{C}$, then it is also reachable from $s_0$ by some path in $\mathcal{C}|_{U=u_0}$ that never changes $u$.

\begin{theorem}
\label{theorem:detailed_balanced_clamping}
Consider a detailed balanced CRN $\mathcal{C} = (\mathcal{S}, \mathcal{R},k)$ with reachability class $\Omega_{s_0}$ from initial state $s_0$. Partition the species into two disjoint sets $U=\{U_1,\dots,U_{M_u}\} \subset \mathcal{S} \textrm{ and } Y=\{Y_1, \dots, Y_{M_y}\} \subset \mathcal{S}$ with $M_u+M_y=M=|\mathcal{S}|$. Let the projection of $s_0$ onto $U$ and $Y$ be $u_0$ and $y_0$. The conditional distribution $\Prob{(y\mid u)}$ implied by the stationary distribution $\pi$ of $\mathcal{C}$ is equivalent to the stationary distribution of a clamped CRN, $\mathcal{C}|_{U=u}$ starting from initial state $s_0$ with $u_0 = u$, provided that clamping preserves reachability.
\end{theorem}

\begin{proof} 
We have $\mathcal{G}(u, y) = \sum_{i=1}^{M_u} u_i G[U_i]+\log (u_i !) + \sum_{i=1}^{M_y} y_i G[Y_i]+\log (y_i !)$.  
Let the reachability class of $\mathcal{C}|_{U=u}$ be $\Gamma_{s_0}$, its projection onto $Y$ be $\Gamma_{s_0}^Y$, and $\Omega^Y_{s_0|U=u_0} = \{ y \textrm{ s.t. } (u_0,y) \in \Omega_{s_0} \}$ with $\Omega^Y_{s_0|U=u_0} = \Gamma_{s_0}^Y$.
Then, the conditional probability distribution of the unclamped CRN is given by
\begin{equation}
\label{clamping theorem eq 1}
\Prob{(y\mid u)} =
\frac{\pi(u,y)}{\sum_{y^\prime \in \Gamma^Y_{s_0}} \pi(u, y^\prime)} =
\frac{\e^{-\mathcal{G}(u,y)}}
{\sum_{y^\prime \in \Gamma^Y_{s_0}} \e^{-\mathcal{G}(u, y^\prime)}} \ .
\end{equation}
Simply removing pairs of forwards and backwards reactions will preserve detailed balance for unaffected transitions, and hence the clamped system remains a detailed balanced CRN with the same free energy function. We then  readily see that the clamped CRN's stationary distribution, $\pi_{c}(y|u)$ is given by
\begin{equation}
\label{clamping theorem eq 2}
\pi_{c}(y|u) = 
\frac{\e^{-\mathcal{G}(u, y)}}{Z_c(u)}
\quad \textrm{ with }
\quad Z_c(u) = \sum_{y' \in \Gamma^Y_{s_0}} \e^{-\mathcal{G}(u, y')}\ . \quad \QED
\end{equation} 

The original CRN and the clamped CRN do not need to have the same initial conditions as long as the initial conditions have the same reachability classes. However, even if the two CRNs have the same initial conditions, it is possible that the clamping process will make some part of $\Omega^Y_{s_0|U=u_0}$ inaccessible to $C|_{U=u}$, in which case this theorem will not hold. 

\end{proof}

\begin{theorem}
Assume the reachability class of a detailed balanced CRN can be expressed as the product of subspaces, $\Omega_{s_0} = \prod_{j=1}^L \Omega_{s_0}^j$. Then the steady-state distributions of each subspace will be independent for each product space: $\pi(s) = \prod_{j=1}^L \pi^j(s^j)$, where $s = (s^1,\dots,s^L)$ and $\pi^j$ is the distribution over $\Omega_{s_0}^j$.
\label{theorem:poisson2}
\end{theorem}
\begin{proof} 
If $\Omega_{s_0}$ is decomposable into a product of subspaces $\Omega^j_{s_0}$, with $j = 1...L$, then each subspace involves disjoint sets of species $Y^j = \{S^j_1, \dots, S^j_{M_j}\}$. In this case the steady-state distribution of a detailed balanced CRN can be factorized due to the simple nature of $\mathcal{G}(s)$ given by eq. (\ref{eq:dbcrn3}):
\begin{equation}
\pi(s) = \frac{\e^{-\mathcal{G}(s)}}{Z} = \frac{\prod_{j=1}^L\e^{-\mathcal{G}(s^j)}}{\prod_{j=1}^L\sum_{{s^j}^\prime \in \Omega^j_{s_0}}  \e^{-\mathcal{G}({s^j}^\prime)}} = \prod_{j=1}^L \frac{\e^{-\mathcal{G}(s^j)}}{ \sum_{{s^j}^\prime \in \Omega^j_{s_0}}  \e^{-\mathcal{G}({s^j}^\prime)}}
= \prod_{j=1}^L \frac{\e^{-\mathcal{G}(s^j)}}{Z^j} \ ,
\end{equation}
where $s^j = (s_1^j, s_2^j,\dots, s^j_{M_j})$ is  state of the set of species within the subspace $j$.$\QED$
\end{proof}

The product form $\pi(s)$ means that species from separate subspaces $\Omega^j_{s_0}$ are statistically independent. To develop non-trivial conditional probabilities for the states of different species, therefore, it is necessary either to use a non-detailed balanced CRN by driving the system out of equilibrium, or to generate complex interdependencies through conservation laws that constrain reachability classes and ``entangle'' the state spaces for different species. We explore both of these possibilities in the following sections.


\label{sec:constructions1}
\label{sec:constructions}
\vspace{-.2cm}
\subsection{Direct Implementation of a Chemical Boltzmann Machine (DCBM):}
\vspace{-.1cm}

We first consider the most direct way to implement an $N$-node Boltzmann machine with a chemical system. Recall that a BM has a state space $\Omega_{BM} = \{0, 1\}^N$ and an energy function
$E(x_1,x_2,\dots,x_N)=-\sum_{i<j}w_{ij}x_ix_j - \sum_i \theta_i x_i$. 
We use a dual rail representation of each node $i$ by two CRN species $X_i^{ON}$ and $X_i^{OFF}$ and a conservation law, $x_i^{ON}+x_i^{OFF}=1$. The species $X_i^{ON}$ and $X_i^{OFF}$ could represent activation states of an enzyme.  The CRN has $M=2N$ species and states
$s=(x_1^{ON}, x_1^{OFF}, \dots, x_N^{ON}, x_N^{OFF})$. 
Although there are $2^{2N}$ states in which each species has a count of at most one, only $1/2^N$ of these states are reachable due to the conservation laws.  Let $\Omega_{DCBM}$ be the states reachable from a valid initial state. There exists a one-to-one invertible mapping $\mathcal{F}: \Omega_{BM} \to \Omega_{DCBM}$ which maps the states $x \in \Omega_{BM}$ of a BM to states $s = \mathcal{F}(x) \in \Omega_{DCBM}$ of the CBM, according to $x_i^{ON} = x_i$ and $x_i^{OFF} = 1-x_i$.

\begin{figure}[t!!!]
\begin{center}
\includegraphics[width=0.85\textwidth]{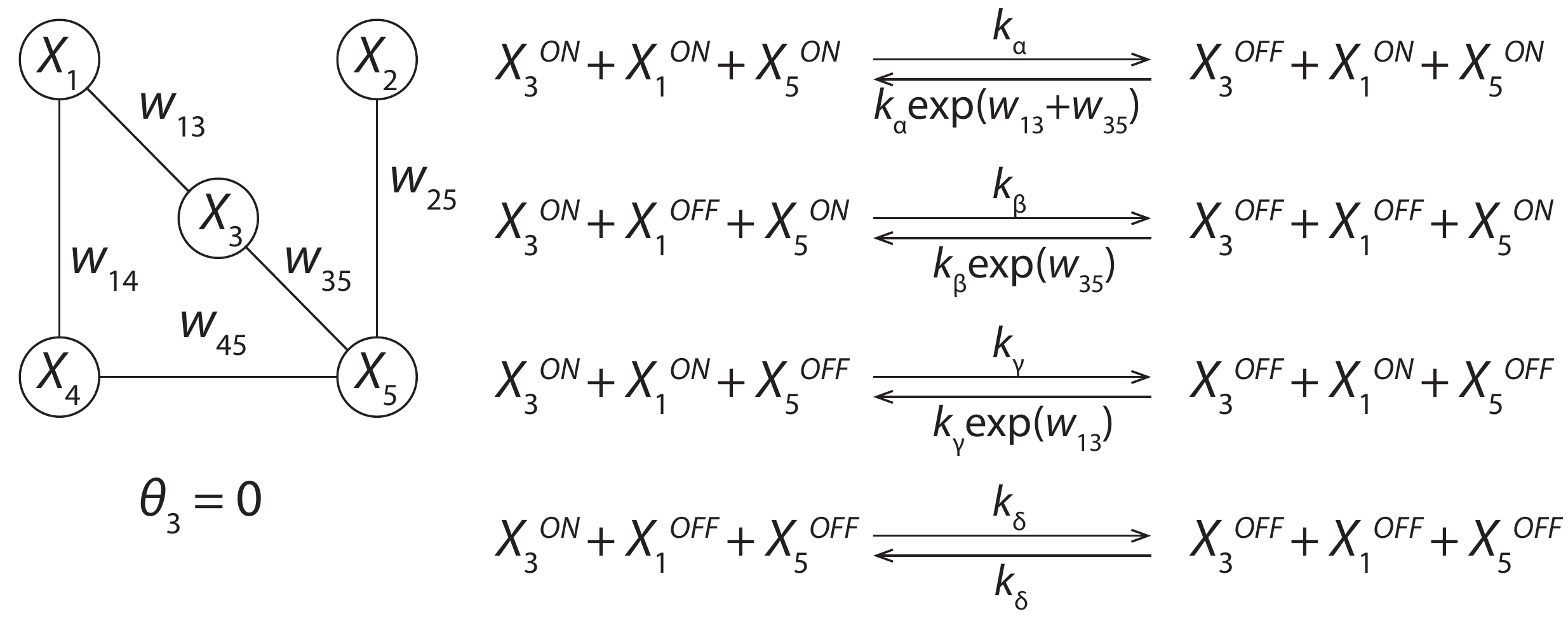}
\caption{The reactions required by the dynamics of a single node using the direct CBM implementation. We consider a simple network with the illustrated topology, and display the required reactions for node 3. Since node 3 has degree 2, there are 4 possible states of its neighbors, and hence four distinct pairs of reactions for the species of node 3. The relative rates of each pair of reactions is set by $w_{ij}$ as indicated (where, for simplicity, we have assumed $\theta_3=0$).
\label{fig:construction1}
}
\end{center}
\end{figure}

Reactions are intended to provide a continuous-time analog of the typical BM implementations,  such as the Gibbs sampling method discussed in Section~\ref{sec:BM}. In each reaction $r$, only the species $X_i^{ON}$ and $X_i^{OFF}$, corresponding to a single node $i$, change ($\nu_r - \mu_r$ has one non-zero component). 
To reproduce the stationary distribution of a Boltzmann machine with energy function $E(x)$, it is sufficient to require that the CTMC for the CRN satisfies
\begin{equation}
\label{direct_cbm_eq1}
s \xrightleftharpoons[]{} s' 
\quad \textrm{with} \quad 
\frac{R_{s \to s'}}{R_{s' \to s}} 
= \frac{\e^{-E(s')}}{\e^{-E(s)}}
= \e^{\theta_i + \sum_{j \in \N{i}} w_{ij} x_j^{ON}}
\end{equation}
where $s$ is any state with $x_i^{OFF}=1$, and $s'$ has $x_i^{ON}=1$ but is otherwise the same.
Such a choice would enforce detailed balance of the CTMC, with the desired steady-state distribution 
\begin{equation}
\pi(s) = \frac{1}{Z}\e^{-E(s)} = \frac{1}{Z} \e^{ -\sum_{i<j}w_{ij}x_i^{ON} x_j^{ON} - \sum_i \theta_i x_i^{ON}} \ . 
\end{equation}

To implement such a CRN, we define a reaction set $\mathcal{R}$ that contains a distinct pair of reactions for each possible state of the neighbors of $i$ for which $w_{ij}\neq 0$. Let $\alpha^i \in \{ ON, OFF \}^{ | \N{i} | }$ denote a state of neighboring species. Then, the necessary reactions and rate constants are
\begin{equation}
X_i^{ON} + \sum_{j \in \N{i}} 
X_j^{\alpha^i_j} \xrightleftharpoons[k_{i^+|\alpha_i}]{k_{i^-|\alpha_i}} X_i^{OFF} +  \sum_{j \in \N{i}}
X_j^{\alpha^i_j}, \quad  \frac{k_{i^+\mid \alpha^i}}{k_{i^-\mid \alpha^i}} = \e^{\theta_i + \sum_{j \in \N{i}} w_{ij} x_j^{ON}},
\label{direct_cbm_rxn}
\end{equation}
for each $i$ and every possible state $\alpha$. 
In physical terms, the species representing the neighbors of node $i$ collectively catalyze $X_i^{OFF} \leftrightharpoons X_i^{ON}$, with a separate pair of reactions for every possible $\alpha^i$.  
While this entails a large number of reactions ($2^{ | \N{i} | + 1}$ for each node $i$), it allows the rate constants for each configuration of neighbors to be distinct, and thus to satisfy the ratio of rate constants given in (\ref{direct_cbm_rxn}). For CRN states that satisfy the conservation laws $x_i^{ON}+x_i^{OFF}=1$, there will be a unique reaction that can flip any given bit, and thus the CTMC detailed balance (\ref{direct_cbm_eq1}) also holds, yielding the correct $\pi(s)$.
The construction is illustrated by example in Figure~\ref{fig:construction1} and compared to other constructions in Figure~\ref{Fig_ExactConstructions}.

The distribution $\pi(s)$ is identical to that of the BM, both with and without clamping. Reachability is preserved by clamping, as all states satisfying the conservation laws and clamping can be reached in the clamped CRN.  All results derived for traditional BMs therefore apply, including conditional inference and the Hebbian learning rule. The construction can be generalized to any graphical model and indeed to any finite Markov chain defined on a positive integer lattice. 

With the DCBM, we have shown that CRNs can express the same distributions as BMs, and are thus very expressive. However, since each possible state $\alpha^i$ of $\N{i}$ is associated with two reactions, the number of reactions of the CRN is exponentially large in the typical node degree $d$ in the original BM. Moreover, the scheme requires high molecularity reactions in which multiple catalysts effect a single transition (the molecularity grows linearly with $d$). Physical implementations are therefore likely to be challenging. We further note that as a consequence of Theorem~\ref{theorem:poisson2}, the DCBM cannot be detailed balanced at the level of the underlying chemistry, due to its simple conservation laws. Physically, this means that the DCBM must use a fuel species to drive each reaction. Details of this argument are given in the appendix (\ref{appendix:dcbm_fuel}).

\vspace{-.2cm}
\subsection{The Edge Species CBM Construction (ECBM):}\label{ss:edgespeciescbm}
\vspace{-.1cm}

Can a detailed balanced CRN also implement a Boltzmann machine, or is it necessary to break detailed balance at the level of the CRN reactions, as in the DCBM? Here we show that it is not necessary by introducing a detailed balanced CRN that uses species to represent both the nodes and edges of a BM. The $N$ nodes of a BM are converted into $N$ pairs of species, $X_i^{ON} \textrm{ and } X_i^{OFF}$, via a dual rail implementation identical to that used in the DCBM. Similarly, the edges $w_{ij}$ are represented by dual rail edge species $W_{ij}^{ON}$ and $W_{ij}^{OFF}$ with the conservation law $w_{ij}^{ON}+w_{ij}^{OFF}=1$ for $1 \leq j < i \leq N$. Note that we may slightly abuse notation and let $W_{ij}^{\alpha_{ij}}$ and $W_{ji}^{\alpha_{ij}}$, with $\alpha_{ij} \in \{ON, OFF \}$, represent the same chemical species.

To have detailed balance, we associate energies to each node species determined by the bias in a BM, $G[X^{ON}_i]= -\theta_i$ and $G[X^{OFF}_i] = 0$. Similarly, each edge species has an energy determined by the corresponding edge weight $G[W^{ON}_{ij}] = -w_{ij}$ and $G[W^{OFF}_{ij}] = 0$.
Finally, we define a set of catalytic reactions that ensure that the states of edge and node species are consistent, meaning $w_{ij}^{ON}=1$ if and only if $x_i^{ON} = 1 \textrm{ and } x_j^{ON}=1$. To achieve this coupling, the reactions that switch node $i$ are always catalyzed by the species corresponding to the set of neighboring nodes $\N{i}$. Simultaneously, these reactions switch edge ${ij}$ if $j \in \N{i}$ and $x_j^{ON}=1$, maintaining  $x_i^{ON}x_j^{ON} = w_{ij}^{ON}$.
The set of reactions that result from this scheme are
\begin{equation}
\label{EdgeSpeciesBCRN}
X_i^{OFF} + \sum_{j \in \N{i}} \! X_j^{\alpha_j^i} + \!\!\! \sum_{j \in \N{i},\,x_j^{ON}=1} \!\! W_{ij}^{OFF} \rightleftharpoons
X_i^{ON} + \sum_{j \in \N{i}} \! X_j^{\alpha_j^i} + \!\!\! \sum_{j \in \N{i},\,x_j^{ON}=1} \!\! W_{ij}^{ON}.
\end{equation}
This reaction scheme is visualized in Figure~\ref{Fig_ExactConstructions}. Just like in the DCBM, there is a separate pair of reactions for each node $i$ for each state of its neighbors $\alpha^i$. In this case, however, the backwards reaction in (\ref{EdgeSpeciesBCRN}) does represent a transition that is a true chemical inversion of the forwards reaction.  Further,
given a valid initial state, clamping any subset of the $X_i^{ON/OFF}$ species preserves reachability.

\begin{figure}[tb!!!]
\centering
\includegraphics[width=0.75\textwidth]{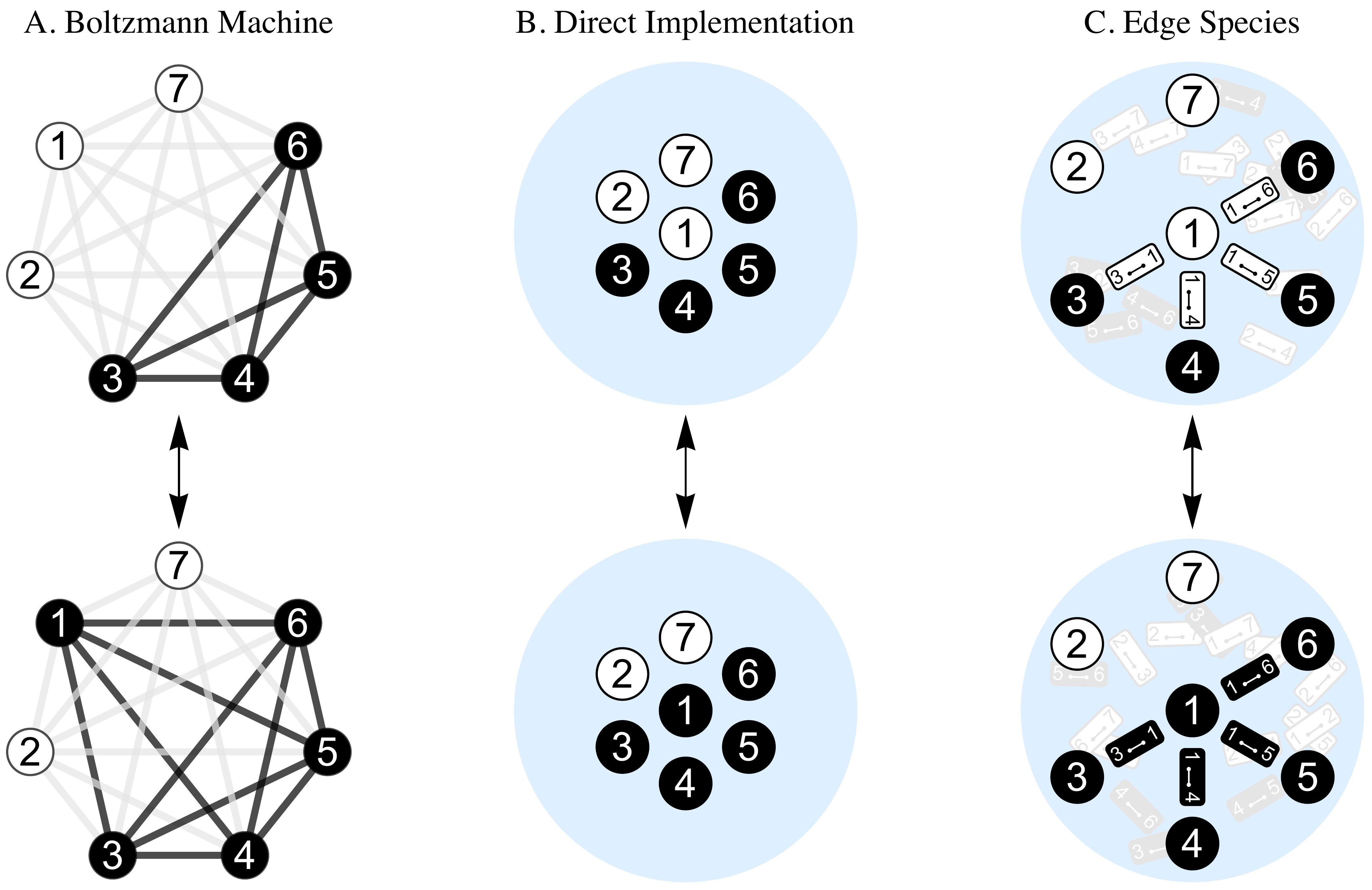}
\caption{
\label{Fig_ExactConstructions} 
Comparison of the switching of a node in exact constructions for fully-connected topologies. Black circles indicate ON species (or nodes), and white circles indicate OFF species. Similarly, black/white rectangles indicate ON/OFF edge species. Species not involved in the reaction have been grayed out. \textbf{A.} A Boltzmann machine. Black edges contribute to the energy function. \textbf{B.} The direct implementation of a chemical Boltzmann machine. All species jointly catalyze the conversion of $X_1^{\rm OFF}$ to $X_1^{\rm ON}$. \textbf{C.} The edge species chemical Boltzmann machine. $X_1^{\rm OFF}$ is converted to ON simultaneously with $W_{14}^{OFF},\,W_{15}^{OFF},\,W_{16}^{OFF}$ and $W_{17}^{OFF}$; all other species involved act as catalysts. 
}
\end{figure}

\begin{theorem}
\label{thrm:ecbm_equivalence}
The stationary distribution $\pi(x^{ON}, x^{OFF}, w^{ON}, w^{OFF})$ of the\linebreak ECBM is equivalent to the stationary distribution of a Boltzmann machine, $\Prob(x)$, provided that the ECBM begins in a valid state obeying $w_{ij}^{ON} = x_i^{ON} x_j^{ON}$ and one applies a one-to-one invertible mapping $\mathcal{F}$ between BM and ECBM states, as described below. 
\end{theorem}

\begin{proof}
If this CRN begins in a consistent state, then every subsequent reaction will conserve this condition. The combined conservation laws $x_i^{ON}+x_i^{OFF}=1$, $w^{ON}_{ij}+w^{OFF}_{ij} =1$, and $w^{ON}_{ij}=x^{ON}_i x^{ON}_j$ ensure that the set of values $x_i^{ON}$ uniquely determine the CRN state for the ECBM, and thus
--- similar to how the BM and DCBM states were related ---
we can define a one-to-one invertible mapping $\mathcal{F}$ that sets
$x_i^{ON}=x_i$ and obeys the conservation laws.

The ECBM is detailed balanced and therefore its stationary distribution has the form (\ref{ProductPoissonThrm}). Substituting the conservation law $w_{ij} = x_i x_j$ and omitting species with 0 energy results in 
\begin{equation}
\pi(x^{ON}, x^{OFF}, w^{ON}, w^{OFF}) = \frac{1}{Z_\pi} \e^{-\sum_{i \neq j}G[W_{ij}^{ON}]x_i^{ON}x_j^{ON} -\sum_i G[X_i^{ON}]x_i^{ON}}
\end{equation} 
Comparing this expression to the distribution of a BM, equation (\ref{BM Dist}), the above expressions are equivalent provided that their partition functions are equivalent. To see this is the case, notice that: 1) the partition function is just a sum over the Gibbs factors across the entire state space. 2) The Gibbs factors take the same form between the ECBM and BM (as shown above). And 3) the reachable state spaces spaces are equivalent. Thus a sum over all possible Gibbs factors will be equal. Therefore, $Z_{BM} = Z_\pi$ and the theorem is proven. $\QED$
\end{proof}

Via the ECBM, we have shown that even detailed balanced CRNs can represent rich distributions and are able to calculate conditional distributions through clamping as proven in Theorem~\ref{theorem:detailed_balanced_clamping}. Due to being detailed balanced, this construction requires no fuel molecules and performs sampling via the intrinsic equilibrium fluctuations of the CRN. Moreover, it is only necessary to tune molecular energies in this construction, since appropriate relative rate constants follow by definition.
This construction is possible due to the complex set of conservation laws that ensure that the reachability classes of all the $X_i^{ON/OFF}$ species are tightly coupled via the $W_{ij}^{ON/OFF}$ species. One implication is that this construction does not generalize easily to non-binary species counts. Additionally, issues related to high molecularity reactions and large number of reactions remain. 

\vspace{-.2cm}
\section{Approximate Bimolecular Implementations}
\vspace{-.2cm}

The DCBM and the ECBM both require reactions of high molecularity. High molecularity reactions and systems involving many species are physically challenging to implement and also potentially suffer from long mixing times. In this section, we discuss an approximation scheme to create CBMs with lower molecularity reactions and thus overcome these issues.

\vspace{-.2cm}
\subsection{Taylor Series Chemical Boltzmann Machine (TCBM):}
\vspace{-.1cm}

Here, we demonstrate a compact CBM that approximates a BM. It is not detailed balanced on either the Markov chain or the CRN level, but uses only $2N$ species and unimolecular and bimolecular reactions. The TCBM is a non-equilibrium CBM of the kind discussed in Section~\ref{sec:constructions1} that uses a dual-rail representation and single-node transitions to approximately implement a BM. The reactions are given by:
\begin{align}
\label{TCBM Reactions}
X_i^{OFF} & \xrightleftharpoons[k]{k} X_i^{ON} \nonumber \\
X_j^{ON}+ X_i^{OFF} & \xrightarrow{ka_{ij}} X_j^{ON} + X_i^{ON} \nonumber \\
 X_j^{ON}+X_i^{ON}  & \xrightarrow{kb_{ij}} X_j^{ON} + X_i^{OFF} 
\end{align}
which, with appropriate initial conditions, preserve the conservation law that $x_i^{ON}+x_i^{OFF}=1$.

This model's parameters can be taken directly from the weights of a BM, $w_{ij}$. First, define a symmetric matrix $W$. Decompose this matrix into the difference of two positive matrices, $W = A - B$, where $a_{ij} \in A$ are all $w_{ij}>0$ and $b_{ij} \in B$ are the absolute values of all $w_{ij} < 0$. Finally, $k$ is an arbitrary overall rate.  This construction can be understood as an approximation of equation (\ref{direct_cbm_eq1}), which dictates that for two states
$s$ and $s'$ that differ only in bit $i$ with $x_i^{ON}=1$ in state $s'$, the CTMC transition rates must satisfy
\begin{equation}
\label{TaylorExpansion}
\frac{R_{s \to s'}}{R_{s' \to s}} = 
\frac{\e^{\sum_{j \neq i} a_{ij} x_j^{ON}}}{\e^{\sum_{j \neq i} b_{ij} x_j^{ON}}} = \frac{1+\sum_{j \neq i} a_{ij} x_j^{ON} + \mathcal{O}((\sum_{j \neq i} a_{ij} x_j^{ON}))^2}{1+\sum_{j \neq i} b_{ij} x_j^{ON} + \mathcal{O}((\sum_{j \neq i} b_{ij} x_j^{ON}))^2} \ ,
\end{equation}
The bias $\theta_i$ has been absorbed into $w_{ij}$ for notational clarity by assuming there is some $x_0^{ON} = 1$ whose weights act as biases. The TCBM is a bimolecular CRN obeying the same conservation laws as the DCBM in which each species $j$ acts as an independent catalyst for transitions in $i$ with reaction rates determined by $a_{ij}$ and $b_{ij}$. The relative propensities of this network are exactly equal to the linear expansion of the relative propensities shown in the last term in (\ref{TaylorExpansion}). Specifically, the numerator is the sum of the reaction propensities for $X_i^{OFF} \rightarrow X_i^{ON}$ and the denominator is the sum of the reaction propensities for $X_i^{ON} \rightarrow X_i^{OFF}$, in each case plus a constant term due to the unimolecular reactions. We thus propose the simple  scheme in (\ref{TCBM Reactions}) as an approximate CBM; Figure~\ref{Fig_CBMInference}A depicts this TCBM schematically. This model bears some resemblance to protein phosphorylation networks where adding or removing a phosphate group is analogous to turning a species on or off; both are driven, catalytic processes capable of diverse computation.

\vspace{-.2cm}
\subsection{Approximate BCRN Inference:}
\vspace{-.1cm}

Remarkably, this simple approximate CBM can reasonably approximate the inferential capabilities of a BM. We demonstrate this by using (\ref{TCBM Reactions}) to convert a BM trained on the MNIST dataset \cite{LeCun1998} to a TCBM (Figure~\ref{Fig_CBMInference}). 
We then compare the BM and the TCBM side by side. Digit classification is shown in Figure~\ref{Fig_CBMInference} panels C and D for a BM and in Figure~\ref{Fig_CBMInference} panels F and G for a CBM as confusion heatmaps. Classification is carried out by clamping the image nodes to MNIST images and averaging the values of the classification nodes. As  is apparent from these plots, the BM does a fairly reasonable job classifying these digits, but struggles on the number 5. The CBM functions as a very noisy version of the BM with nodes in general much more likely to be on. The CBM has also faithfully inherited the capabilities and limitations of the BM and similarly struggles to classify the digit 5. Digit generation is shown in Figure~\ref{Fig_CBMInference}E for a BM and \ref{Fig_CBMInference}H for a CBM. Generation was carried out by clamping a single class node to 1 and all other class nodes to 0, then averaging the output of the image nodes after the network had equilibrated. For each generated image, we show the raw output and the top 85th percentile of nodes, a thresholding which helps visualize the noisy output. As is apparent from the raw output, the CBM approximation scheme does not generate images nearly as distinctly as the BM. However, this approximation does faithfully reproduce plausible digits when filtered for the top 85th percentile. Additional training and simulation details can be found in the appendix (\ref{appendix:sim_details})

The overall performance of the CRN is reasonable, given the fact that weights were simply imported from a BM without re-optimization. The TCBM only approximates the distribution implied by these weights and, in the absence of detailed balance, does not have an established formal relationship between clamping and conditioning.

\begin{figure}[tb!!!]

\includegraphics[width=1.0\textwidth]{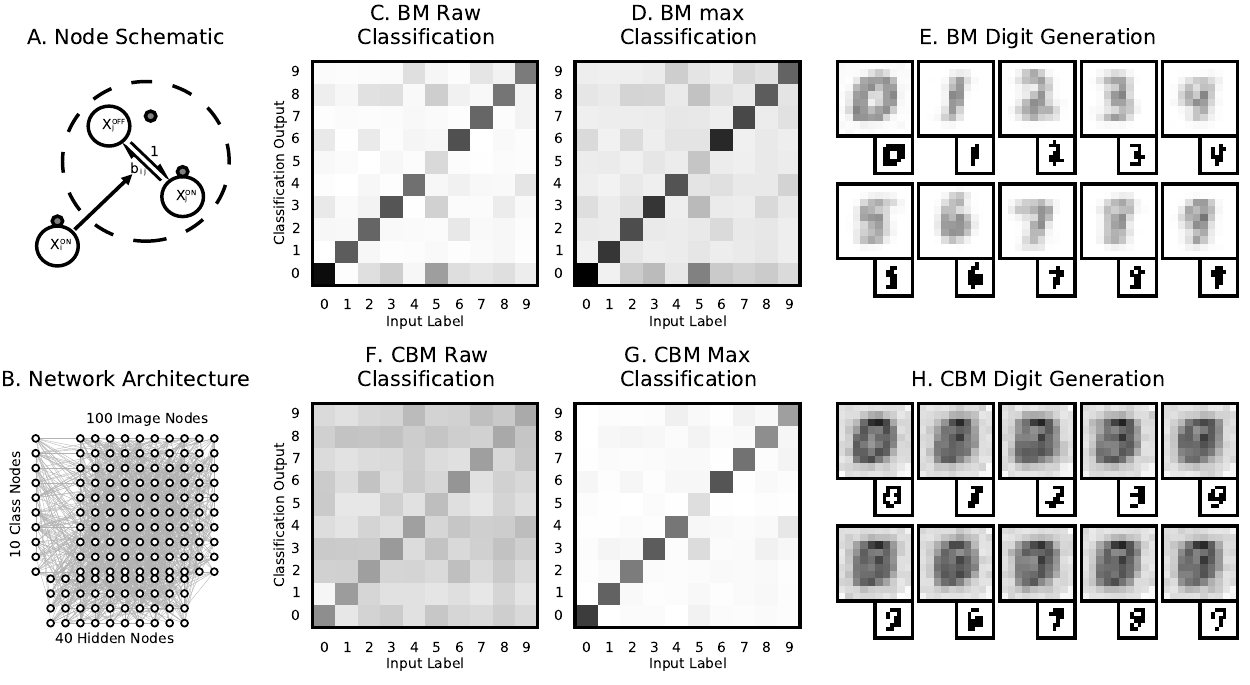}
\caption{
\textbf{A.} CRN underlying an individual node of the TCBM approximation. In this case a negative weight, $w_{ij}<0$ is shown because $X_i^{ON}$ catalyzes $X_j^{ON} \to X_j^{OFF}$.
\textbf{B.} Network architecture used for simulations is fully connected but only 10 percent of edges are shown for clarity. 
\textbf{C.} Average raw classification output of a BM running with clamped MNIST digits.
\textbf{D.} Average max classification output of a BM running with clamped MNIST digits. 
\textbf{E.} Digits generated by a BM by clamping individual class nodes. Small sub-boxes in the bottom corners are plots of the top 85th percentile of pixels. 
\textbf{F.} Average raw classification output of a TCBM running with clamped MNIST digits. 
\textbf{G.} Average max classification output of a TCBM running with clamped MNIST digits. 
\textbf{H.} Digits generated by a TCBM by clamping individual class nodes.  Small sub-boxes in the bottom corners are plots of the top 85th percentile of pixels. 
\label{Fig_CBMInference}
}
\end{figure}

\vspace{-.2cm}
\section{Detailed Balanced CRN Learning Rule}
\vspace{-.2cm}

A broad class of detailed balanced chemical reaction networks can be trained with a Hebbian learning rule between a waking phase (clamped) and sleeping phase (free) that is reminiscent of the classic gradient descent learning algorithm for a BM~\cite{Hinton1984,Ackley1985}. A similar algorithm for detailed balanced CRNs follows.

First we state a simple case of Theorem~\ref{thrm:detailed_balance_learning_rule_with_hidden_units} where we just want a CRN with stationary distribution $\pi$ over $\Omega_{s_0}$ to learn a target distribution $Q$ also defined on $\Omega_{s_0}$. Then, the learning rule is given by
\begin{equation}
\label{visible_only_learning_rule}
\frac{d g_i}{dt} = -\frac{\partial D_{KL}}{\partial g_i} = \langle s_i \rangle_{Q} -\langle s_i \rangle_{\pi}. 
\end{equation}
Here, $\langle s_i \rangle_\pi$ and $\langle s_i \rangle_Q$ denote the expected count of the species $S_i$ with respect to the probability distributions $\pi$ and $Q$, respectively, and $g_i = G[S_i]$ is the energy of species $S_i$. Theorem~\ref{thrm:detailed_balance_learning_rule_with_hidden_units} generalizes this procedure to cases where there are hidden species.



\begin{theorem}
\label{thrm:detailed_balance_learning_rule_with_hidden_units}
Let $\mathcal{C} = (\mathcal{S}, \mathcal{R}, k)$ be a detailed balanced chemical reaction network with stationary distribution $\pi(s)$ on $\Omega_{s_0}$. Consider a partition $(V,H)$ of the set $\mathcal{S}$ of species into visible and hidden species such that $\pi(s) = \pi(v, h)$. Require that for all visible states $v$, the clamped CRN $\mathcal{C}|_{V=v}$ preserves reachability. Let $Q(v)>0 \textrm{ for all } v \in \Omega^V_{s_0} = \{v \textrm{ s.t. } (v, h) \in \Omega_{s_0}\}$ be a target distribution defined on the projection of $\Omega_{s_0}$ onto $V$. Furthermore, let $\pi_Q(v, h) =Q(v) \pi(h \mid v)$ be the weighted mixture of stationary distributions of the clamped CRNs $\mathcal{C} |_{V=v}$ with $v$ drawn from the distribution $Q$. Then, the gradient of the Kullback-Leibler divergence from $\pi_V$ to $Q$ with respect to the energy, $g_i = G[S_i]$, of the species $S_i$ is given by
\begin{equation}
\frac{\partial D_{KL}(Q || \pi_V)}{\partial g_i} =
\langle s_i \rangle_{\pi} -\langle s_i \rangle_{\pi_Q}
\label{eq_db_lr_thrm2}
\end{equation}
where $\pi_V(v) = \sum_{h \in \Omega_{s_0}^H} \pi(v, h)$ is the marginal $\pi(v, h)$ over hidden species $H$.
\end{theorem}
\begin{proof}
Applying Theorem~\ref{theorem:detailed_balanced_clamping}, the clamped CRN ensemble $\pi_Q(s)$ may be written
\begin{equation}
\pi_Q(s) = \pi_Q(v, h) = Q(v) \pi(h \mid v) = Q(v) \frac{\pi(v, h)}{\sum_{h \in \Omega_{s_0}^H} \pi(v, h)} = Q(v) \frac{\pi(v, h)}{\pi_V(v)} \ .
\end{equation}
Additionally we will need the partial derivative of a Gibbs factor and the partition function with respect to $g_i$,
\begin{equation}
\frac{\partial \e^{-\mathcal{G}(s)}}{\partial g_i} = -s_i \e^{-\mathcal{G}(s)} \quad \textrm{and} \quad \frac{\partial Z}{\partial g_i} = -Z \langle s_i \rangle_{\pi} \ \ .
\end{equation}
Using these results, the partial derivative of any detailed balanced CRN's distribution at a particular state $s$, with respect to an energy $g_i$, is
\begin{equation}
\label{Learning Rule Eq 7}
\frac{\partial \pi(s)}{\partial g_i} = \frac{\partial}{\partial g_i}
\frac{1}{Z} \e^{-\mathcal{G}(s)} =  \langle  s_i \rangle_\pi \pi(s) - s_i \pi(s) \ .
\end{equation}
Noting that $Q$ has no dependence on $g_i$, the gradient of the Kullback-Leibler divergence can then be written,
\begin{align*}
\frac{\partial D_{KL}(Q || \pi_V)}{\partial g_i} 
&= \frac{\partial}{\partial g_i} \sum_{v \in \Omega_{s_0}^V} Q(v) \log \frac{Q(v)}{\pi_V(v)} = \sum_{v \in \Omega_{s_0}^V}  -\frac{Q(v)}{\pi_V(v)} \frac{\partial \pi_V(v)}{\partial g_i} \\
&= -\sum_{v \in \Omega_{s_0}^V} \sum_{h \in \Omega_{s_0}^H} \frac{Q(v)}{\pi_V(v)} \pi(v, h) \langle s_i \rangle_\pi  -  \frac{Q(v)}{\pi_V(v)}\pi(v, h) s_i \\
&= -\sum_{(v, h) \in \Omega_{s_0}} \pi_Q(v, h) \langle s_i \rangle_\pi - \pi_Q(v, h) s_i = -\langle s_i \rangle_\pi + \langle s_i \rangle_{\pi_Q} \QED
\end{align*}
\end{proof}
In the special case where there are no hidden species, which is to say the target distribution $Q$ is defined over the whole reachability class $\Omega_{s_0}$, then $\pi_V(v) = \pi(s)$ and $\pi_Q(s) = Q(s)$ and the gradient has the simple form shown in equation (\ref{visible_only_learning_rule}).


Applying gradient descent via $\frac{dg_i}{dt} = -\frac{\partial D_{KL}}{\partial g_i}$, we thus have a simple {\em in silico} training algorithm to train any detailed balanced CRN such that it minimizes the Kullback-Leibler divergence from $\pi_V$ to $Q$. If $H = \emptyset$, simulate the CRN freely to estimate the average counts $\langle s_i \rangle$ under $\pi(s)$. Then compare to the average counts under the target $Q(s)$ and update the species' energies accordingly. If $H \neq \emptyset$, clamp the visible species to some $v \in \Omega_{s_0}^V$ with probability $Q(v)$
and simulate the dynamics of the hidden units. Repeat to sample an ensemble of clamped CRNs $\mathcal{C}|_{V=q}$. Because clamping $v$ preserves reachability, Gillespie simulations of the CRN with the $V$ species clamped to the data values $v$ will sample appropriately. This gives the average counts under $\pi_Q$.

This CBM learning rule is more general than the classical Boltzmann machine learning rule, as it applies to arbitrary detailed balanced CRNs, including those with arbitrary conservation laws and arbitrarily large species counts (but still subject to the constraint that reachability under clamping must be preserved).  That said, at first glance the CBM learning rule appears weaker than the classical Boltzmann machine learning rule, as it depends exclusively on mean values $\langle s_i \rangle$, whereas the Boltzmann machine learning rule relies primarily on second-order correlations $\langle x_i x_j \rangle$.  In fact, though, conservation laws within the CRN can effectively transform mean values into higher-order correlations.  A case in point would be to apply the CBM learning rule to the ECBM network:  For $g_i = G[X_i^{ON}] = -\theta_i$, $\frac{d \theta_i}{dt} = -\frac{d g_i}{dt} = \langle x_i^{ON}\rangle_{\pi_Q} - \langle x_i^{ON} \rangle_{\pi}$, and for $g_i = G[W_{ij}^{ON}] = -w_{ij}$, $\frac{d w_{ij}}{dt} = -\frac{d g_i}{dt} = \langle w_{ij}^{ON}\rangle_{\pi_Q} - \langle w_{ij}^{ON} \rangle_{\pi} = \langle x_i^{ON} x_j^{ON} \rangle_{\pi_Q} - \langle x_i^{ON} x_j^{ON} \rangle_{\pi}$,
which exactly matches the classical Boltzmann machine learning rule if we assert that the energies of $OFF$ species are fixed at zero.

\vspace{-.2cm}
\section{Discussion}
\vspace{-.2cm}

We have given one approximate and two exact constructions that allow CRNs to function as Boltzmann machines. BMs are a ``gold standard" generative model capable of performing numerous  computational tasks and approximating a wide range of distributions. Our constructions demonstrate that CRNs have the same computational power as a BM. In particular, CRNs can produce the same class of distributions and can compute conditional probabilities via the clamping process. Moreover, the TCBM construction appears similar in architecture to protein phosphorylation networks. Both models are non-equilibrium (i.e., require a fuel source) and make use of molecules that have an on/off (e.g., phosphorylated/unphosphorylated) state. Additionally, there are clear  similarities between our exact schemes and combinatorial regulation of genetic networks by transcription factors. In this case, both models make use of combinatoric networks of detailed-balanced interactions (e.g., binding/unbinding) to catalyze a state change in a molecule (e.g., turning a gene on/off). We note that our constructions differ from some biological counterparts in requiring binary molecular counts. However, in some cases we believe that biology may make use of conservation laws (such as having only a single copy of a gene) to allow for chemical networks networks to perform low-cost computations. In the future, we plan to examine these cases in a biological setting as well as generalize our models to higher counts.

Developing these CBMs leads us to an important distinction between equilibrium, detailed-balanced CRNs with steady state distributions determined by molecular energies, and CRNs that do not obey detailed balance in the underlying chemistry. The second category includes those that nonetheless appear detailed balanced at the Markov chain level. Physically, this distinction is especially important: a non-detailed balanced CRN will always require some kind of implicit fuel molecule (maintained by a chemostat) to run and the steady state will not be an equilibrium steady state due to the continuous driving from the fuel molecules. 
A detailed balanced CRN (at the chemical level) requires no fuel molecules: and thus \emph{the chemical circuit can act as a sampler without fuel cost}. Despite this advantage, working with detailed balanced CRNs presents additional challenges: to ensure that chemical species do not have independent distributions, species counts must be carefully coupled via conservation laws. 

\begin{table}[tb!!!]
\begin{center}
\begin{tabular}[h!]{| c | c | c | c | c | c |}

\hline
\textbf{Model} & \textbf{Species} & \textbf{Reactions} & \textbf{Molecularity}& \textbf{Detailed Balance}\\
\hline
Direct CBM 	& $2N$ 		& $N2^{d+1}$ 		& $d+1$ 		& CTMC  \\
\hline
Edge CBM & $2N + dN$ & $N2^{d+1}$ & $\leq 2d+1$ & CRN and CTMC\\
\hline
Taylor CBM & $2N$ & $2N+dN$ & $\leq 2$ & Neither\\
\hline

\end{tabular}
\caption{The complexity and underlying properties of our constructions for reproducing a BM with $N$ nodes of degree $d$. Detailed balance describes whether the construction is detailed balanced at the CRN level, at the CTMC level, or neither. \label{DiscussionTable} }
\end{center}
\end{table}


Our constructions also highlight important  complexity issues underlying CBM design. The number of species, the number of reactions, and the reaction molecularity needed to implement a particular BM are relevant. Trade-offs appear to arise between these different factors, and the thermodynamic requirements of a given design. A breakdown of the main features of each CBM is given in Table~\ref{DiscussionTable}. Summarizing, the TCBM is by far the simplest construction, using $\mathcal{O}(N)$ species, at most $\mathcal{O}(N^2)$ reactions, with molecularity $\leq 2$. However, this happens at the expense of not being an exact recreation of a BM, and the requirement of a continuous consumption of fuel molecules. The DCBM is the next simplest in complexity terms, using $\mathcal{O}(N)$ species, $\mathcal{O}(N2^N)$ reactions, and molecularity of at most $N$. Like the TCBM, the DCBM requires fuel molecules because it is not detailed balanced at the CRN level. The ECBM is considerably more complex than the DCBM, using quadratically more species, $\mathcal{O}(N^2)$, the same number of reactions, $\mathcal{O}(N2^N)$ and double the reaction molecularity. The ECBM makes up for this increased complexity by being detailed balanced at the CRN level, meaning that it functions in equilibrium without implicit fuel species. 

Finally, we have shown that a broad class of detailed balanced CRNs can be trained using a Hebbian learning rule between a waking phase (clamped) and sleeping phase (free) reminiscent of the gradient descent algorithm for a BM. This exciting finding allows for straightforward optimization of detailed balanced CRNs' distributions.

This work provides a foundation for future investigations of probabilistic molecular computation. In particular, how more general restrictions on reachability classes can generate other ``interesting" distributions in detailed balanced CRNs remains an exciting question. We further wonder if the learning rule algorithm can be generalized to certain classes of non-detailed balanced CRNs. Another immediate question is whether these ideas can be generalized to non-binary molecular counts. From a physical standpoint, plausible implementations of the clamping process and the energetic and thermodynamic constraints require investigation. Indeed, a more realistic understanding of how a CBM might be implemented physically will help us identify when these kinds of inferential computations are being performed in real biological systems and could lead to building a synthetic CBM. 

\vspace{.1cm}
{\bf Acknowledgements.}
This work was supported in part by U.S. National Science Foundation (NSF) graduate fellowships to WP and to AOM, by NSF grant CCF-1317694 to EW, and by the Gordon and Betty Moore Foundation through Grant GBMF2809 to the
Caltech Programmable Molecular Technology Initiative (PMTI).

\bibliography{Bib1}
\bibliographystyle{unsrt}

\section{Appendix}

\vspace{-.2cm}
\subsection{Application of Theorem~\ref{theorem:poisson2}: The Direct CBM Must Use Implicit Fuel Species}
\vspace{-.1cm}
\label{appendix:dcbm_fuel}

Here, we use Theorem~\ref{theorem:poisson2} to analyze the direct implementation of a CBM and show that it cannot be detailed balanced and thereby requires implicit fuel molecules. First, notice that the the conservation laws used in this construction are of a simple form. The states accessible by $(X_i^{ON},X_i^{OFF})$ are independent of $(X_j^{ON},X_j^{OFF})$ for $i \neq j$, and therefore the reachability class is a product over the subspaces of  each individual node. As a consequence, by Theorem~\ref{theorem:poisson2}, the system must be out of equilibrium and violate detailed balance at the  level of the CRN because, by construction, this system is equivalent to a BM and has correlations between nodes $i$ and $j$ whenever $w_{ij} \neq 0$. In physical terms, the presence of catalysts cannot influence the equilibrium yield of a species, and therefore a circuit which uses catalysis to bias distributions of species must be powered by a supply of chemical fuel molecules \cite{Qian2007,Beard2008,Ouldridge2017}. It is also worth noting that, as a consequence, this scheme cannot be implemented by tuning of (free) energies; it is fundamentally necessary to carefully tune all of the rate constants individually (via implicit fuel molecules) to ensure that detailed balance is maintained at the level of the Markov chain for the species of interest.

\vspace{-.2cm}
\subsection{BM Training and TCBM Simulation Details:}
\vspace{-.1cm}
\label{appendix:sim_details}

We trained a BM using stochastic gradient descent on the MNIST dataset, down sampled to be 10 pixels by 10 pixels \cite{LeCun1998}. The BM has 100 visible image units (representing a 10 x 10 image), 10 visible class nodes, and 40 hidden nodes as depicted in Figure~\ref{Fig_CBMInference}B. Our training data consisted of the concatenation of down sampled MNIST images and their classes projected onto the 10 class nodes. The weights and biases of the trained BM were converted to reaction rates for a CBM using the Taylor series approximation. This CBM consists of 300 species, 300 unimolecular reactions and  22350 bimolecular reactions. 
The resulting CBM was then compared side-by-side with the trained BM on image classification and generation. The BM was simulated using custom Gibbs sampling written in Python. The CRN was simulated on a custom Stochastic Simulation Algorithm (SSA)~\cite{gillespie2007stochastic} algorithm written in Cython. All simulations, including network training, were run locally on a notebook or on a single high performance Amazon Cloud server.

Classification was carried out on all 10000 MNIST validation images using both the BM and the CBM. Each 10 by 10 gray-scale image was converted to a binary sample image by comparing the gray-scale image's pixels (which are represented as real numbers between 0 and 1) to a uniform distribution over the same range. The network's image units were then clamped to the binary sample and the hidden units and class units were allowed to reach steady state. This process was carried out 3 times for each MNIST validation image, resulting in 30000 sample images being classified. Raw classification scores were computed by averaging the class nodes' outputs for 20000 simulation steps after 20000 steps of burn-in (Gibbs sampling for the BM, SSA for the CBM). Max classification was computed by taking the most probable class from the raw classification output. Raw classification and max classification confusion heatmaps, showing the average classification across all sample images as a function of the true label are shown in Figure~\ref{Fig_CBMInference} panels C and D for a BM and in Figure~\ref{Fig_CBMInference} panels F and G for a CBM. 

Image generation was carried out by clamping the class nodes with a single class, $0...9$, taking the value of 1 and all other classes being 0, and then allowing the network to reach steady state. Generated images were computed by averaging the image nodes over 50000 simulation steps (Gibbs sampling for the BM, SSA for the CBM) after 25000 steps of burn-in. Generation results are shown in Figure~\ref{Fig_CBMInference}E for a BM and Figure~\ref{Fig_CBMInference}H for a CBM.

\end{document}